\documentclass{neus2025}

\usepackage{amsfonts}
\usepackage{textcomp}
\usepackage{tikz}
\newcommand{\reals}{\mathbb{R}}
\newcommand{\nats}{\mathbb{N}}
\newcommand{\identity}{\mathbb{I}}
\usepackage{times}
\usepackage{appendix}

\usepackage{subcaption}
\usepackage[rightcaption]{sidecap}

\title[$k$-NBCs for Unknown Nonlinear Dynamics]{$k$-Inductive Neural Barrier Certificates for Unknown Nonlinear Dynamics}

\author{%
 \Name{Ben Wooding} \Email{ben.wooding@vanderbilt.edu}\\
 \Name{Hongchao Zhang} \Email{hongchao.zhang@vanderbilt.edu}\\
 \Name{Taylor T Johnson} \Email{taylor.johnson@vanderbilt.edu}\\
 \addr Vanderbilt University, Nashville, TN, USA
 \AND
 \Name{Abolfazl Lavaei} \Email{abolfazl.lavaei@newcastle.ac.uk}\\
 \addr Newcastle University, Newcastle upon Tyne, UK%
}

\begin{document}

\maketitle

\begin{abstract}
  While conventional ($k=1$) discrete-time barrier certificate conditions impose strict safety constraints by requiring the function to be non-increasing at every step, $k$-inductive barrier certificates relax this by allowing a temporary increase—up to $k-1$ times, each within a threshold $\epsilon$—while maintaining overall safety, and improving flexibility. This paper leverages neural networks and constructs $k$-inductive neural barrier certificates ($k$-NBCs) for (partially) unknown nonlinear systems. While neural networks offer scalability in the design process, they lack formal guarantees, requiring additional approaches such as counterexample-guided inductive synthesis (CEGIS) with satisfiability modulo theories (SMT) for verification. However, the CEGIS-SMT framework requires knowledge of system dynamics, which is unavailable in practical settings. To address this, we leverage the generalization of the Willems \textit{et al.}'s fundamental lemma, using a \emph{single state trajectory}, to construct a data-driven representation of {(partially)} unknown models for SMT verification without sacrificing accuracy. Additionally, CEGIS-SMT further removes the constraint of restricting barrier certificates to specific function classes, such as sum-of-squares, enabling greater flexibility in their design. We validate our approach on three nonlinear case studies with {(partially)} unknown dynamics.
\end{abstract}


\section{Introduction}
\label{sec:introduction}
Neural barrier certificates (NBCs) are growing in popularity as a technique to provide safety over dynamical systems~\citep{peruffo2021automated,prajna2004safety,qin2021learning}. NBCs are \textit{neuro-symbolic}~\citep{garcez2023neurosymbolic}, integrating neural approaches (neural networks, deep learning, \emph{etc.}) with symbolic approaches (logic, rules, explicit reasoning, \emph{etc.}). One approach to design NBCs is a two-step closed-loop hierarchy: (i) leveraging a neural network as a \emph{function approximator} to generate a candidate that satisfies desired system properties using a collection of data samples, and  (ii) verification that the candidate satisfies the property across the whole state space using \emph{a known model} and satisfiability modulo theories (SMT) solvers~\citep{z3,dreal}. If the second step fails, counterexamples refine the candidate NBC from the first phase, and the loop iterates until a candidate passes the second phase or the process is manually terminated. This loop is known as \emph{counterexample-guided inductive synthesis (CEGIS)}.

A key strength of designing NBCs with neural networks is their ability to generate \emph{general nonlinear candidates}, unlike the more common \emph{sum-of-squares} approach, which is restricted to polynomial functions~\citep{SOSTOOLS,parrilo2020sum}. Moreover, verifying an NBC candidate using SMT solvers is computationally more efficient than directly searching for a solution via SMT. However, a limitation of SMT-based NBCs is their \emph{dependence on prior knowledge of system dynamics} to guarantee validity on unsampled points.

Data-driven methods offer a promising alternative to deriving safety certificates when the system model is unavailable. Two widely used techniques are \emph{the scenario approach}~\citep{calafiore2006scenario,campi2018introduction,kazemi2024data} and the \emph{single-trajectory approach}, which builds on Willems \textit{et al.}'s fundamental lemma~\citep{willems2005note, de2019formulas,verhoek2021fundamental}. The scenario approach gives probabilistic guarantees that strengthen with more data, while the single-trajectory approach leverages a \emph{persistently excited} state trajectory to infer system behavior. In this work, we adopt the latter to extend the CEGIS-SMT framework for NBCs, removing the assumption of a  known model \emph{a priori}.

This work focuses on discrete-time dynamical systems and introduces a more general class of NBCs, termed $k$-inductive neural barrier certificates ($k$-NBCs), by leveraging $k$-induction~\citep{Anand2021kSafety}. Specifically, we address system \emph{safety}, ensuring that a system initialized in a predefined region never transitions into an unsafe zone. The $k$-induction framework relaxes the strict non-increasing constraint of barrier certificates~\citep{prajna2004safety}, allowing temporary increases over a finite horizon $k$ before enforcing a decrease, as illustrated in Figure~\ref{fig:k-BC}.

\begin{SCfigure}[1.2][t]
    \begin{tikzpicture}[scale=0.8]
    \node at (1.8,2) {$X$};
    \draw[thick] (0,0) ellipse (4.5cm and 2.5cm);
    \draw[thick,fill=blue!20, blue,opacity=0.4] (-1.5,0) ellipse (1.5cm and 1.5cm);
    \node at (-1.5,1) {$X_\mathcal{I}$};
    \node at (-1.5,0.5) {$\mathcal{B}(x) \leq 0$};
    \draw[thick,fill=red!20, red,opacity=0.4] (2.7,0) ellipse (1cm and 1cm);
    \node at (2.7,0.5) {$X_\mathcal{U}$};
    \node at (2.7,0) {$\mathcal{B}(x) > \lambda$};
    \draw[thick, red] (-1,0) ellipse (2.5cm and 2cm);
    \draw[thick, blue,dashed] (-1,0) ellipse (2cm and 1.8cm);
    \draw[thick, black, ->] plot[smooth] coordinates { 
        (-0.5,0.8) 
        (0.2,0.6) 
        (0.8,0.3) 
        (1.2,-0.1) 
        (0.6,-0.8) 
        (1,-1.1) 
        (-0.3,-1.5) 
        (-1.1,-0.4) 
        (-1.6,-0.8) 
        (-1.5,0)
    };
    \end{tikzpicture}
    \caption{Safety illustration: The black trajectory, starting in the initial region $X_{\mathcal{I}} \subset X$, may temporarily increase (crossing the blue-dashed 0-level set) as long as it decreases within $k$ steps and never exceeds the red level set $\lambda = (k-1)\epsilon$. This ensures it does not reach the unsafe region $X_{\mathcal{U}} \subset X$, where $X_{\mathcal{I}} \cap X_{\mathcal{U}} = \emptyset$.}
    \label{fig:k-BC}\vspace{-0.3cm}
\end{SCfigure}
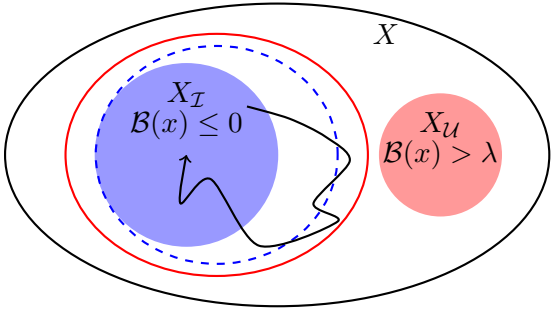

\subsection{Contributions}
Our contributions are twofold:
\begin{itemize}
    \item We construct a data-driven closed-loop model of {(partially)} unknown system dynamics using the generalization of Willems \emph{et al.}'s fundamental lemma. Since CEGIS-SMT requires a known model for verifying candidate neural barrier certificates, our approach eliminates the need for an a priori closed-form system representation, enabling verification with only a single state trajectory.
    \item We significantly relax the safety conditions with $k$-NBCs; removing the non-increasing requirement, and enabling \emph{general} nonlinear barrier certificates for general nonlinear {(partially)} \emph{unknown} dynamical systems.
\end{itemize} 

\subsection{Related Work}
Prior work has explored the integration of data-driven methods with barrier certificates. Studies such as~\citep{nejati2023formal,salamati2024data,anand2023formally} employ barrier certificates for verification using the scenario approach, while~\citep{nejati2022data,samari2024singletrajectory} leverage Willems \emph{et al.}'s fundamental lemma for controller synthesis. However, these approaches do not incorporate $k$-induction or \emph{neural} barrier certificates. $k$-inductive barrier certificates have been studied in~\citep{murali2022scenario} using the scenario and sum-of-squares approaches, but only for linear systems. The closest work to ours is~\citep{wooding2024learning}, which designs $k$-inductive control barrier certificates for nonlinear dynamics using sum-of-squares. However, that approach relies on\emph{ nonlinearity cancellation through control input}, making verification of nonlinear systems infeasible. Furthermore, this work extends $k$-NBCs beyond the quadratic barrier certificate constraints imposed in previous approaches, allowing greater flexibility in their design.

NBCs have been explored in~\citep{yang2023model}, but in a non-formal manner, achieving ``near-zero" constraint violations through reinforcement learning. In~\citep{zhao2020synthesizing,NEURIPS2024_b7868ded}, NBCs are synthesized using SMT solvers for controller design but without the CEGIS approach. Both works~\citep{yang2023model,zhao2020synthesizing} rely on a known system model and focus on conventional NBCs rather than the relaxed $k$-inductive version. In contrast,\citep{zhou2022neural,hsieh2025certifying} address black-box systems, eliminating the need for a model in neural approaches. However,\citep{zhou2022neural} assumes knowledge of Lipschitz constants, while~\citep{hsieh2025certifying} employs CEGIS without SMT solvers and establishes termination guarantees. Notably, both focus on stability rather than barrier certificates, as considered in this work. Other neural-certificate based works include~\citep{debauche2025formal,abate2020formal,mathiesen2022safety} and the survey of~\cite{dawson2023safe}. Several software tools have been developed in this area: \textsf{FOSSIL}~\citep{abate2021fossil,edwards2024fossil} for NBCs, \textsf{PRoTECT}~\citep{wooding2024protect} for model-based barrier certificates, and \textsf{TRUST}~\citep{gardner2025trust} for data-driven barrier certificates.

\subsection{Organization}
The paper is structured as follows: Section~\ref{sec:problem} defines the {(partially)} \emph{unknown} system under consideration. Section~\ref{sec:data-drivenmodels} presents how the generalization of Willems \emph{et al.}'s fundamental lemma is used to represent a single state trajectory as a system model. Section~\ref{sec:k-induction} introduces data-driven $k$-inductive barrier certificate conditions for nonlinear systems. Section~\ref{sec:nbc} details the use of the CEGIS-SMT framework to generate and verify candidate $k$-NBCs for {(partially)} unknown systems. Section~\ref{sec:case-studies} demonstrates the results on a deliberately challenging \emph{highly nonlinear} case study, followed by conclusions in Section~\ref{sec:conclusion}. Due to space constraints, proofs and additional (physical) case studies are presented in the appendix.

\subsection{Notation}
 We use $\reals,$ $\reals_{\geq 0}$, $\reals_{>0}$, $\nats$ and $\nats_{>0}$ to denote the real numbers, non-negative real numbers, positive real numbers, and the non-negative and positive integers, respectively. We define an $(n \times n)$ identity matrix by  $\identity_n$. For a state space $X$, we define $X_{\mathcal{I}}\subset X$ and $X_{\mathcal{U}} \subset X$ as the set of initial states and set of unsafe states, respectively, where $X_{\mathcal{I}} \cap X_{\mathcal{U}} = \emptyset$.

\section{Problem Formulation}
\label{sec:problem}

\begin{definition}
	\label{def:system-description-NS}
	A discrete-time nonlinear system (dt-NS) in this work is described by
	\begin{equation}\label{eq:dt-NS1}
		\Sigma: x^+ = f(x),    \end{equation}
	where $x^+$ represents the state variables at the next time step, i.e., $x^+ := x(k + 1), \; k \in \mathbb{N}$, $x\in X$ is a system state with $X\subset\mathbb{R}^n$ being the state set, and $f: X \to X$ is an unknown transition map.
\end{definition}

For the sake of simplicity, one can reformulate the system in \eqref{eq:dt-NS1} as \begin{equation}
	\label{eq:dt-NS}
	\Sigma: x^+ = A\mathfrak{D}(x),
\end{equation}
where $A\in\reals^{n\times N}$ is a constant matrix,  and $\mathfrak{D}(x)\in\reals^N$ is a vector of nonlinear terms in state $x$. We denote by $x_{x_0}(k)$ the state trajectory of $\Sigma$ at time $k\in\nats$ starting from an initial condition $x_0=x(0)$.

While we assume that the matrix $A$ is \emph{unknown}, which reflects practical scenarios, a dictionary $\mathfrak{D}$ is assumed to be available, constructed to adequately capture the true dynamics by being \emph{sufficiently comprehensive} to encompass all possible terms in the actual system, even if some irrelevant terms are included. We have thus far used the term \emph{(partially) unknown} to emphasize the availability of an extensive dictionary together with unknown parameters; however, for simplicity of presentation, we henceforth use the term \emph{unknown}, as the distinction should now be clear from the context. We also note that there are some related studies in the literature that allow $\mathfrak{D}$ to exclude certain nonlinearities  by assuming boundedness of the neglected terms with known bounds; see, \emph{e.g.,}~\citep{de2023cancellation,samari2025data}.

Existing works on neural barrier certificates often assume a \emph{{fully} known} closed-form system model~\eqref{def:system-description-NS}, which enables validation through SMT solvers. In the next section, we relax this assumption by adopting a data-driven approach.

\section{Data-Driven Models}
\label{sec:data-drivenmodels}

We collect state data from the \emph{unknown} dt-NS over the time horizon $[0,T]$, where $T\in\mathbb{N}_{>0}$ is the number of collected samples:
\begin{align}
	\label{eq:state-data}
	&\mathcal{X}_{0,T} = [x(0),x(1),\ldots,x(T-1)]\in\reals^{n\times T},\\
	\label{eq:derivative-data}
	&\mathcal{X}_{1,T} = [x(1),x(2),\ldots,x(T)]\in\reals^{n\times T}.
\end{align}
We also consider the following data-driven representation of the extensive dictionary of nonlinear functions:
\begin{align}
	\label{eq:dictionary-data}
	\mathcal{D}_{0,T} = [\mathfrak{D}(x(0)),\mathfrak{D}(x(1)),\ldots,\mathfrak{D}(x(T-1))],
\end{align}
where $\mathcal{D}_{0,T}\in\reals^{N\times T}$. While the data is collected noise-free, the current approach can be extended to noisy data by building upon the results of~\citep{guo2021data}.

We first consider the following lemma for the dt-NS system, to obtain a data-based representation of a closed-loop dt-NS (with the continuous-time equivalent by~\citet{nejati2022data}).

\begin{lemma}
	\label{lem:non-Q-matrix}
	Let $Q(x)$ be a $(T\times N)$ matrix such that
	\begin{equation}
		\label{eq:non-Q-matrix-X0}
		\mathbb{I}_N = \mathcal{D}_{0,T}Q(x),
	\end{equation}
	where $\mathcal{D}_{0,T}$ is an $(N\times T)$ full row-rank matrix. The closed-loop system $x^+=A\mathfrak{D}(x)$ has the following data-based representation:
	\begin{align}
		\label{eq:closed-loop-dtNS}
		x^+ = \mathcal{X}_{1,T}Q(x)\mathfrak{D}(x).
	\end{align}
\end{lemma}

\begin{remark}
	To enforce $\mathcal{D}_{0,T}$ to be full row-rank, the number of samples $T$ should be at least $N$~\citep{de2023cancellation} as a necessary condition. Since this matrix is derived based on sampled data, this rank condition is readily satisfiable.
\end{remark}

Since the $k$-step evolution for dt-NS is required in our setting, we describe it using a recursive expression in the following theorem.

\begin{lemma}\label{lem:recursive-dt-NS}
	Under Lemma  \ref{lem:non-Q-matrix}, the data-based evolution of a dt-NS after $k$ steps is described by
	\begin{align}
		\label{eq:non-closed-loop-k}
		x^{k+}=f_k(x) = \mathcal{X}_{1,T}Q(f_{k-1}(x))\mathfrak{D}(f_{k-1}(x)),
	\end{align}
	where $f_1(x) = \mathcal{X}_{1,T}Q(x)\mathfrak{D}(x)$.
\end{lemma}

\section{$k$-Inductive Safety Barrier Certificates}
\label{sec:k-induction}

We define safety, as illustrated in Figure~\ref{fig:k-BC}, by the following:

\begin{definition}[Safety]
    A dt-NS $\Sigma$ is safe if all trajectories $x_{x_0}(k)$ evolving from each state in the initial region $x_0\in X_\mathcal{I}$ never eventually reach the unsafe region $X_\mathcal{U}$, \emph{i.e.}, $x_{x_0}(k)\notin X_{\mathcal U},~\forall k\in\nats$. 
\end{definition}

We consider $k$-inductive safety barrier certificates ($k$-BCs) as defined in~\citep{Anand2021kSafety}. Importantly, we \textbf{do not} impose any predefined structure on the $k$-BC $\mathcal{B}(x)$, such as assuming it to be a quadratic polynomial of the form $x^\top Px$, a limitation of prior data-driven work~\citep{wooding2024learning}.

\subsection{General $k$-BC Conditions}
Building on Lemma~\ref{lem:non-Q-matrix} and Lemma~\ref{lem:recursive-dt-NS}, we establish the $k$-BC conditions for data-driven safety guarantees of an unknown dt-NS in the following proposition, adapting the $k$-BC definition from~\citep{Anand2021kSafety}.

\begin{proposition}
	\label{thm:data_k-inductive-CBC}
	Given a dt-NS $\Sigma$ in Definition~\ref{def:system-description-NS}, and data-driven system evolutions~\eqref{eq:closed-loop-dtNS} and~\eqref{eq:non-closed-loop-k}, a function $\mathcal{B}: X\rightarrow \reals$ is  a $k$-inductive barrier certificate ($k$-BC) for $\Sigma$ if there exist $k\in\nats_{>0}$ and $\epsilon\in\reals_{\geq0}$, such that
	\begin{subequations}
		\begin{align}
			\label{eq:safe_kCBC_cond1}
			\mathcal{B}(x)& \leq 0, &\forall x\in X_{\mathcal I}, \\
			\label{eq:safe_kCBC_cond2}
			\mathcal{B}(x)& > (k-1)\epsilon, &\forall x\in X_{\mathcal U},\\
			\label{eq:safe_kCBC_cond3}
			\mathcal{B}\big(\mathcal{X}_{1,T}Q(x)\mathfrak{D}(x)\big) &\leq \mathcal{B}(x) + \epsilon, \quad &\forall x\in X~\text{where}~\mathcal{B}(x) \leq (k-1)\epsilon, \\
			\label{eq:safe_kCBC_cond4} \mathcal{B}\big(\mathcal{X}_{1,T}Q(f_{k-1}(x))\mathfrak{D}(f_{k-1}(x))\big) &\leq \mathcal{B}(x), \quad &\forall x\in X~\text{where}~\mathcal{B}(x) \leq 0,
		\end{align}
		with $f_1(x) = \mathcal{X}_{1,T}Q(x)\mathfrak{D}(x)$ and $f_k(x) = \mathcal{X}_{1,T}Q(f_{k-1}(x))\mathfrak{D}(f_{k-1}(x))$.
	\end{subequations}
\end{proposition}

When $k=1$, the conditions reduce to those of conventional barrier certificates, implying that $k$-BCs with $k> 1$  are always less restrictive without compromising guarantees. Condition~\eqref{eq:safe_kCBC_cond3} relaxes the one-step evolution requirement from strict non-negativity to allowing an increase of up to $\epsilon$. This relaxation is balanced by condition~\eqref{eq:safe_kCBC_cond4}, which enforces strict non-negativity at $k$ future time steps. The term $(k-1)\epsilon$ in~\eqref{eq:safe_kCBC_cond2} plays a crucial role in preventing excessive relaxation, ensuring the system trajectory does not reach the unsafe region.

For completeness, we present the following theorem, adapted from~\citep[Theorem 3.2]{Anand2021kSafety}, providing safety guarantees for discrete-time dynamical systems by leveraging the notion of $k$-BCs from Proposition~\ref{thm:data_k-inductive-CBC}. 

\begin{theorem}
	\label{thm:k-ind-proof}
	Given a dt-NS $\Sigma$ and some $k\in\nats$, suppose $\mathcal{B}(x)$ is a $k$-BC for $\Sigma$ as in Proposition~\ref{thm:data_k-inductive-CBC}. Then, the state trajectory $x_{x_0}(k)$ remains safe (i.e., $x_{x_0}(k)\notin X_{\mathcal U},~\forall k$) for any $x_0\in X_{\mathcal I}$ with data-based dynamic representation satisfying conditions~\eqref{eq:safe_kCBC_cond3} and~\eqref{eq:safe_kCBC_cond4}.
\end{theorem}

Note that conditions~\eqref{eq:safe_kCBC_cond3}-\eqref{eq:safe_kCBC_cond4} are not enforced over all $x\in X$, as Theorem~\ref{thm:k-ind-proof} ensures forward invariance, system trajectories outside this set can be ignored. The $k=2$ special case has~\eqref{eq:safe_kCBC_cond3}-\eqref{eq:safe_kCBC_cond4} both quantified for $x\in X$ where $\mathcal{B}(x)\leq 0$. Sum-of-squares approaches often impose the evolution condition on all states for computational simplicity, whereas SMT solvers enable this additional relaxation. We also present the theory applied to the less-general linear systems in Appendix~\ref{app:linear-systems}, showing simpler conditions can be utilized.

\section{$k$-Inductive Neural Barrier Certificates}
\label{sec:nbc}

We now present the main contribution of this work. We adapt the CEGIS-SMT framework to generate $k$-inductive neural safety barrier certificates ($k$-NBCs) that relax the assumption of knowing a closed-form model of $\Sigma$. We highlight again that to our knowledge, this is the first work designing and verifying NBCs using CEGIS-SMT entirely from data.

\subsection{CEGIS Procedure}

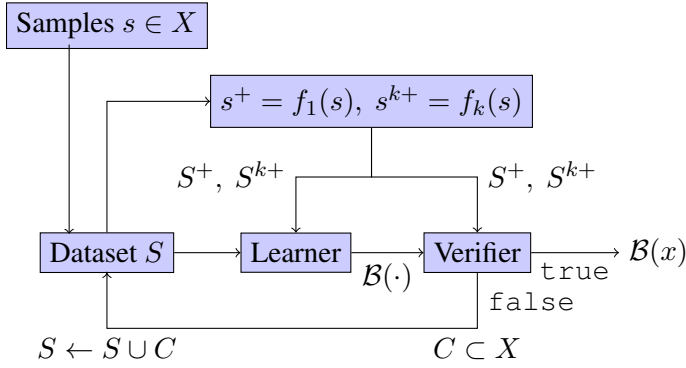
\begin{SCfigure}[1.2][t]
    \begin{tikzpicture}[scale=1.0]
		\node (init) at (-0.5,3) [draw, rectangle,fill=blue!20] {Samples $s\in X$};
		\node (dataset) at (-0.5,0) [draw, rectangle,fill=blue!20] {Dataset~$S$};
		\node (learner) at (2,0) [draw, rectangle,fill=blue!20] {Learner};
		\node (verifier) at (4.4,0) [draw, rectangle,fill=blue!20] {Verifier};
		\node (system) at (3,2) [draw, rectangle,fill=blue!20] {$s^{+} = f_1(s),~s^{k+} = f_k(s)$};
		\node (bx) at (6.8,0){$\mathcal{B}(x)$};
		
		\draw[->] (-1,2.75) -- ++ (0,-2.5);
		\draw[->] (dataset.east) -- (learner.west);
		\draw[->] (dataset.north) |- (system.west);
		\draw[->] (system.south) |-(3,1)-| (learner.north) node[midway, left] {$S^+,~S^{k+}$};
		\draw[->] (system.south) |-(3,1)-| (verifier.north) node[midway, right] {$S^+,~S^{k+}$};
		\draw[->] (learner.east) -- (verifier) node[midway,below] {$\mathcal{B}(\cdot)$};
		\draw[->] (verifier) |- node[near start,right] {\texttt{false}} node[midway, below] {$C\subset X$} (4,-1) -| (-0.5,-1) node[midway, below] {$S \leftarrow S \cup C$} |- (dataset.south);
		\draw[->] (verifier.east) -- (bx.west) node[midway, below] {\texttt{true}};
    \end{tikzpicture}\vspace{0.8cm}
    \caption{The CEGIS procedure for verification with $k$-NBCs. Data points $s\in S$ are sampled from the state space $X$. The \emph{Learner} then generates a candidate $k$-NBC $\mathcal{B}(\cdot)$ based on the dataset triple $(S,~S^+,~S^{k+})$. The \emph{Verifier} checks if the candidate $k$-NBC is valid for all states $x\in X$. If $\mathcal{B}(\cdot)$ is invalid, any counterexamples $c\in C$ are added to the dataset and the Learner generates a new candidate. The process is repeated until a valid $k$-NBC $\mathcal{B}(x)$ is found, or the process is manually terminated.}
    \label{fig:CEGIS}
\end{SCfigure}

We first describe the CEGIS procedure as visualized in Figure~\ref{fig:CEGIS}. At a high level, CEGIS has two main \emph{neuro-symbolic} components: a \emph{learner} (neural component) and a \emph{verifier} (symbolic component). The {learner} generates a candidate $k$-NBC based on the sampled data points. Then the {verifier} checks if the candidate $k$-NBC is valid for the whole state space. If the candidate $k$-NBC is invalid, counterexamples $C$ are added to the dataset and the learner generates a new candidate. This process is repeated until a valid $k$-NBC is found, or the process is manually terminated. We notate $\mathcal{B}(\cdot)$ for a candidate $k$-NBC and $\mathcal{B}(x)$ for a verified $k$-NBC.

\noindent\textbf{Samples and Dataset.} The state space $X$ is sampled for some positive integer $m$ number of data points $s$. The collection of these data points forms the dataset $S$.

\noindent\textbf{System Evolution.}
To verify the conditions of Proposition~\ref{thm:data_k-inductive-CBC}, it is necessary to evaluate the $1$-step and $k$-step evolutions of the system. For each data point $s\in S$, we acquire the state evolutions $s^+$ and $s^{k+}$ by evolving the system for $1$ or $k$ steps from initial state $s$, respectively. Collectively these datasets form a dataset triple $(S,~S^+,~S^{k+})$ that is used by the learner.

\noindent\textbf{Learner.} 
The learner trains a neural network to approximate the $k$-NBC function. When the loss function is optimally minimized, the resulting function serves as a candidate $k$-NBC. The candidate $k$-NBC $\mathcal{B}(\cdot)$ is trained to satisfy the conditions of Proposition~\ref{thm:data_k-inductive-CBC} for the dataset triple $(S,~S^+,~S^{k+})$. Although the candidate satisfies the conditions of Proposition~\ref{thm:data_k-inductive-CBC} for all $s \in S$, it is not necessarily valid for all states $x \in X$ and therefore requires verification.

\noindent\textbf{Verifier.}
The verifier formally checks if the candidate $\mathcal{B}(\cdot)$ violates the conditions of Proposition~\ref{thm:data_k-inductive-CBC} within the whole domain $X$. If a violation is detected, the verifier generates a set of counterexamples $C \subset X$ that includes at least one counterexample $c$. The set $C$ is then appended to the original dataset $S$ and System Evolution and Learner steps are repeated, encouraging the learner to generate a new candidate $k$-NBC for verification. This loop continues indefinitely until it is either manually terminated or the verifier confirms that no counterexamples exist, thereby proving the candidate $k$-NBC is valid across the entire state space, ensuring $\mathcal{B}(\cdot) = \mathcal{B}(x)$.

Verifying the $k$-CBC candidate over the entire state space $X$ using SMT solvers has required knowledge of the closed-form underlying dynamical system~\eqref{eq:dt-NS1}. By leveraging Lemma~\ref{lem:non-Q-matrix} and Lemma~\ref{lem:recursive-dt-NS}, we relax this requirement and instead use the model~\eqref{eq:closed-loop-dtNS}, which is derived solely from a single data trajectory.

\subsection{Training the Neural Network}

A neural network functions as an approximator by processing inputs through a series of weighted operations, biases, and activation functions to generate the desired output. The first main component of the CEGIS-SMT procedure is the \emph{learner}, which assigns weights and biases to the neural network according to a given loss function. The learner designs a candidate $k$-CBC using the hyper-parameters $d$ and $h_1,\ldots,h_d$, representing the depth and width of the neural network, where every connection between nodes has a weight and a bias. The learner optimizes weights and biases using backpropagation, via the Adam optimizer~\citep{kingma2014adam}, over the dataset triple $(S,~S^+,~S^{k+})$. It refines its training whenever the verifier includes new counterexamples.

The training procedure minimizes a loss function that is designed to satisfy the four conditions of Proposition~\ref{thm:data_k-inductive-CBC}. Conditions~\eqref{eq:safe_kCBC_cond1} and~\eqref{eq:safe_kCBC_cond2} apply only to specific subsets of the state space. Therefore, the relevant data points should be filtered accordingly, denoted as $S_{\mathcal I}$ and $S_{\mathcal U}$, respectively. The loss function $\mathcal{L}$ to minimize is then constructed as: \begin{equation}
	\label{eq:loss_function}
	\mathcal{L} = \mathcal{L}_{\mathcal{I}} + \mathcal{L}_{\mathcal{U}}+\mathcal{L}_1+\mathcal{L}_k,
\end{equation}
where: \begin{align*}
	\mathcal{L}_{\mathcal{I}} & = \frac{1}{\vert S_{\mathcal I}\vert}\sum_{s\in S_{\mathcal I}} \text{ReLU}(\mathcal{B}(s) + \eta_1),\\
	\mathcal{L}_{\mathcal{U}} & = \frac{1}{\vert S_{\mathcal U}\vert}\sum_{s\in S_{\mathcal U}} \text{ReLU}(-\mathcal{B}(s) + (k-1)\epsilon + \eta_2),\\
	\mathcal{L}_1 & = \frac{1}{\vert S\vert}\sum_{s\in S} \text{ReLU}(\mathcal{B}\big(\mathcal{X}_{1,T}Q(s)\mathfrak{D}(s)\big)-\mathcal{B}(s) - \epsilon + \eta_3),\\
	\mathcal{L}_k & = \frac{1}{\vert S\vert}\sum_{s\in S} \text{ReLU}(\mathcal{B}\big(\mathcal{X}_{1,T}Q(f_{k-1}(s))\mathfrak{D}(f_{k-1}(s))\big)-\mathcal{B}(s) + \eta_4),
\end{align*}
with the function ReLU$(\cdot)$ and $\eta_i\geq 0,~i=\{1,\ldots,4\}$, are tuning parameter that may increase the conservativeness of the desired function. Iterations of backpropagation are then performed until the value of the loss function is sufficiently small (ideally zero). Note that the loss functions $\mathcal{L}_1$ and $\mathcal{L}_k$ are applied across the entire state space. This is necessary because, until the function is generated, the states satisfying $\mathcal{B}(x) \leq 0$ cannot be determined.

The loss function may not always converge to zero. As a practical strategy, training under more conservative conditions can lead to a function that meets the requirements of Proposition~\ref{thm:data_k-inductive-CBC}, even if the loss does not reach zero.

\subsection{Verifying the Candidate}
The second component of the CEGIS procedure is the \emph{verifier}, which verifies a candidate is a valid $k$-NBC for all states $x$ in the state space, rather than only the sampled states $s$ as trained. Although a closed-form model of~\eqref{eq:dt-NS1} is unavailable, the single-trajectory approach gives a data-driven model~\eqref{eq:closed-loop-dtNS} that serves as a suitable substitute for the true system, given that the rank condition ensures persistently excited data. Since the structure, weights, biases, and activation functions of the neural network are known, a symbolic expression of the candidate $k$-NBC $\mathcal{B}(\cdot)$ can be extracted.

The goal of the SMT solver is to verify if this candidate is a valid $k$-NBC for the full state space, and if not to generate appropriate counterexamples. This check is performed by negating the conditions of Proposition~\ref{thm:data_k-inductive-CBC} and trying to solve for any solution. As the conditions are negated, a solution to the SMT solver is a generated counterexample, which can be passed to the learner to find a new candidate. If no counterexample can be found, the $k$-NBC $\mathcal{B}(x)$ is verified over the whole state space.

In natural language, we check if for a given state space $X$, any of the conditions in Proposition~\ref{thm:data_k-inductive-CBC} are violated. Using implication, and negating every condition, the conditions can be checked based on the regions where they apply. The two evolution conditions can be combined as they refer to the same set of states. The logical formula we check to find a counterexample can be written formally as follows:
\begin{equation}
	\texttt{cex} = (x\in X)\wedge (\neg\textbf{V}_I\vee\neg\textbf{V}_U\vee\neg\textbf{V}_{e1}\vee\neg\textbf{V}_{e2}),
\end{equation}
where \begin{align*}
	\textbf{V}_I :=&~ x\in X_I \implies \mathcal{B}(x) \leq 0, \quad\quad
	\textbf{V}_U :=~ x\in X_U \implies \mathcal{B}(x) > (k-1)\epsilon,\\
	\textbf{V}_{e1} :=&~ \mathcal{B}(x) \leq (k-1)\epsilon \implies
	(\mathcal{B}(\mathcal{X}_{1,T}Q(x)\mathfrak{D}(x)))\leq \mathcal{B}(x) + \epsilon)\\
    \textbf{V}_{e2} :=&~ \mathcal{B}(x) \leq 0 \implies
	(\mathcal{B}(\mathcal{X}_{1,T}Q(f_{k-1}(x))\mathfrak{D}(f_{k-1}(x)))) \leq \mathcal{B}(x)).
\end{align*}

If the above logical formula is satisfied, the SMT solver returns \texttt{sat} along with the \emph{identified counterexample}. If the solver returns \texttt{unsat}, \emph{no counterexample exists}, confirming that the $k$-NBC is valid for all $x \in X$.

\section{Simulation Results}
\label{sec:case-studies}

We demonstrate the efficacy of our approach through three case studies: a \emph{nonlinear polynomial} system, a \emph{physical pendulum} system, and a \emph{highly nonlinear} system. Due to space constraints, only the highly nonlinear system is presented in the main body of the paper as it is the most complex. Both the nonlinear polynomial system and the physical pendulum system are presented in Appendix~\ref{app:case-studies}.

\subsection{Highly Nonlinear System}

\begin{figure}[t]
	\centering
    \hspace{-1.5cm}
	\begin{minipage}[b]{0.48\linewidth}
		\centering
		\includegraphics[width=1.2\columnwidth]{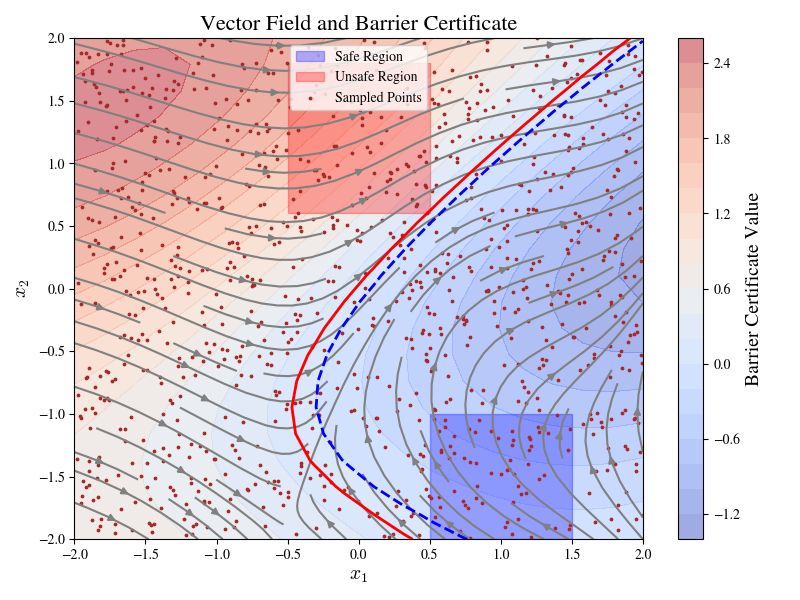}
		\centerline{(a) Initial candidate}
	\end{minipage}
    \hspace{1cm}
	\begin{minipage}[b]{0.48\linewidth}
		\centering
		\includegraphics[width=1.2\columnwidth]{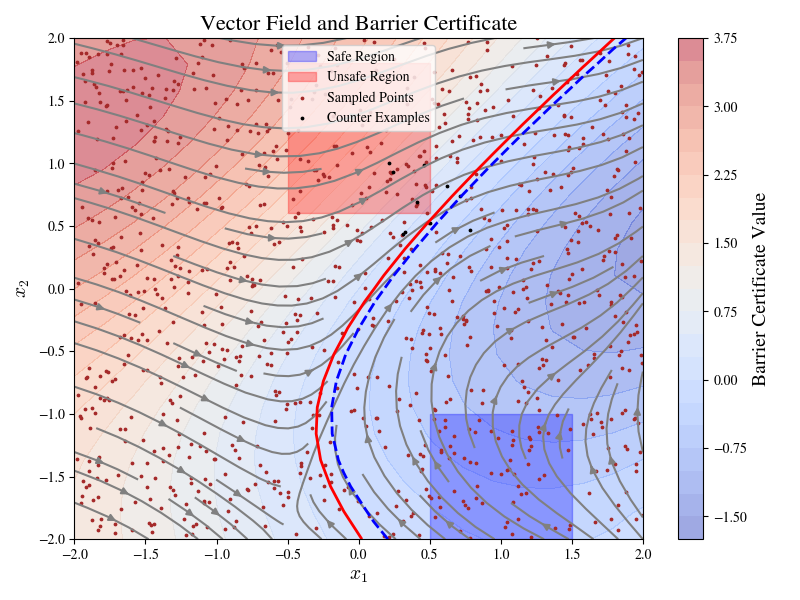}
		\centerline{(b) Verified certificate}
	\end{minipage}
	\caption{$k$-NBC results for the \emph{highly nonlinear system}. (a) \emph{Candidate} $k$-NBC showing violations where $X_{\mathcal{I}}$ and $X_\mathcal{U}$ intersect the level sets. Sampled points are depicted in brown. (b) \emph{Verified} $k$-NBC after retraining with counterexamples (black points). Both show the zero level set in dashed blue, unsafe level set in red, and true vector field in gray.}
	\label{fig:nonpoly2-combined}
\end{figure}

We present our results for the following academic system that is designed to demonstrate the generality of our approach.

\noindent\textbf{System Description.} We consider a discrete-time highly nonlinear system as
\begin{align*}
	\begin{bmatrix}
		x_1^+\\
		x_2^+
	\end{bmatrix} = \begin{bmatrix}
		x_1\\x_2
	\end{bmatrix} +  \Delta t\begin{bmatrix}
		x_2 + e^{-x_1}+\text{sin}^2(x_1)\\
		x_1 -\text{sin}^2(x_1) + \text{cos}^2(x_1)
	\end{bmatrix}\!\!, 
\end{align*}
where the sampling time $\Delta t$ is $0.1$ seconds. The model is assumed to be \emph{unknown}, but an extensive dictionary is available: $\mathfrak{D}(x) = [x_1, x_2, e^{-x_1}, e^{-x_2}, \text{sin}^2(x_1), \text{cos}^2(x_1)]$. We consider a safety specification over the state space $X=[-2,2]^2$ with an initial region $X_\mathcal I = [0.5,1.5]\times[-2,-1]$ and unsafe region $X_\mathcal U = [-0.5,0.5]\times[0.6,1.8]$. As our dictionary has $6$ terms, we consider a single state trajectory of length $7$, from initial state $x_0 = [0.5,-1]$. Using the datasets $\mathcal{X}_{0,T}$, $\mathcal{X}_{1,T}$, and $\mathcal{D}_{0,T}$, we solve~\eqref{eq:non-Q-matrix-X0}, obtaining the data-driven model $x^+ = \mathcal{X}_{1,T}Q(x)\mathfrak{D}(x)$.

\noindent\textbf{Initial Candidate.} We aim to find a $k$-NBC for this system, with $k=2$ and $\epsilon=0.1$. We fix the parameter $\eta_2 = 0.001$. We sample $1000$ data points from the data-driven model to form our dataset triple $(S,S^+,S^{k+})$. We also consider a neural network with two inputs $x_1,x_2$ (one per dimension), one hidden layer containing four nodes, and one output which is the evaluation of our candidate $k$-NBC. Each node $y$ evaluates the formula $y = g(b + \sum_{i=1}^2 x_iw_i)$, where, for each node, $b$ represents the bias, $w_i$ are the weights, and the activation function $g(\cdot)$ is $\text{sin}(\cdot)$ for two hidden nodes and $\text{cos}(\cdot)$ for the other two hidden nodes. For training, backpropagation is performed using the Adam optimizer~\citep{kingma2014adam}, running for $1000$ epochs with a learning rate of $0.1$. Verification with dReal~\citep{dreal} identified a counterexample, with $\delta$-bound set to $0.001$. As seen in Figure~\ref{fig:nonpoly2-combined}(a), the candidate $k$-NBC is invalid, as parts of the region $X_{\mathcal{I}}$ have barrier values above the blue-dashed zero level set. Additionally, the lower right corner of region $X_{\mathcal{U}}$ has $k$-NBC value below the red level set.

\noindent\textbf{Retraining.} From the range of the counterexample, $20$ points within a $0.1$ radius are appended to the dataset for retraining. The model is retrained for another $1000$ epochs with a reduced learning rate of $0.05$. On this iteration, the SMT solver \emph{did not} return a counterexample --- verifying the $k$-NBC $\mathcal{B}(x)$ as valid for all $x$, as seen in Figure~\ref{fig:nonpoly2-combined}(b). All results were generated and verified using the data-driven model. Allowing for rounding, the $k$-NBC is designed  as the following nonlinear non-polynomial function:
\begin{align*}
    \mathcal{B}(x) = &-0.55~\text{sin}(0.54x_1 - 1.32x_2 + 1.14) - 1.35~\text{sin}(0.58x_1 - 0.47x_2 + 0.29) \\&+ 0.65~\text{cos}(0.72x_1 - 0.06x_2 + 1.40) + 0.12~\text{cos}(0.80x_1 - 0.05x_2 + 1.31) + 0.99.
\end{align*}

\begin{remark}
    It naturally follows that using a larger neural network with a broader range of activation functions can further enhance the nonlinearity of $\mathcal{B}(x)$.
\end{remark}

\noindent\textbf{Conventional Barrier Certificate Solution.} We investigate whether the solved $\mathcal{B}(x)$ satisfies the conventional barrier certificate conditions (where $k=1$ and $\epsilon=0$). However, dReal identifies a counterexample within the range $[0.6389,0.6391]\times[-2,-1.9995]$. Consequently, only by applying the $k$-inductive conditions from Proposition~\ref{thm:data_k-inductive-CBC} could we obtain a valid relaxed solution.

As a general observation, across various simulations, we  found that the $k$-NBCs require fewer iterations than conventional NBCs and, in several cases, could identify a valid $k$-NBC on the first attempt.

\section{Conclusion}
\label{sec:conclusion}

We proposed a data-driven method for constructing $k$-inductive neural barrier certificates and verify them using the neuro-symbolic CEGIS-SMT framework, without requiring an \emph{a priori} known model of the underlying system dynamics. By leveraging persistently excited data, a single state trajectory was used to derive a data-driven model of the dynamics, ensuring validity for verification with the SMT solver. Additionally, $k$-induction was employed as a relaxation of conventional discrete-time barrier certificate conditions, and our experiments demonstrated that it can find certificates that conventional barrier certificates fail to design. Neural approaches enabled the discovery of \emph{general} nonlinear barrier certificates for discrete-time systems with \emph{unknown} general nonlinear dynamics. While our paper broadens the applicability of CEGIS-SMT, the existing scalability limitations of the framework remain; nevertheless, the pipeline can be extended by replacing the SMT solver with alternative verification tools (\emph{e.g.,} scenario approach). Future work can consider $t$-barriers, a challenging continuous-time analogy to $k$-NBCs requiring derivative relaxations~\citep{bak2018t}.

\acks{The material presented in this paper is based upon work supported by the National Science Foundation (NSF) through grant number 2220401 and the Defense Advanced Research Projects Agency (DARPA) under contract number FA8750-23-C-0518. Any opinions, findings, and conclusions or recommendations expressed in this paper are those of the authors and do not necessarily reflect the views of DARPA or NSF. Early stages of this work were supported by an EPSRC Doctoral Prize Research Fellowship.}

\bibliography{Bibliography}

\appendix
\section{Proofs}
\label{app:proofs}

\subsection{Proof for Lemma~\ref{lem:non-Q-matrix}}

	The closed-loop dt-NS can be written as
	\begin{align*}
		A\mathfrak{D}(x) = 
		A\overbrace{\mathcal{D}_{0,T}Q(x)}^{\mathbb{I}_N }\mathfrak{D}(x) = \mathcal{X}_{1,T}Q(x)\mathfrak{D}(x),
	\end{align*}
	with $\mathcal{X}_{1,T} = A\mathcal{D}_{0,T}$ as in~\eqref{eq:dictionary-data}. Therefore, $x^+ = \mathcal{X}_{1,T}Q(x)\mathfrak{D}(x)$ is the data-based representation of the closed loop dt-NS, completing the proof.\hfill $\blacksquare$

\subsection{Proof for Lemma~\ref{lem:recursive-dt-NS}}

	Since $x^{k+}$ is obtained by applying $k$ iterations of $A\mathfrak{D}(x)$, it follows that
	\begin{align*}
		x^+ &= A\mathfrak{D}(x) = \mathcal{X}_{1,T}Q(x)\mathfrak{D}(x) = f_1(x), \\
		x^{2+} &= \mathcal{X}_{1,T}Q(x^+)\mathfrak{D}(x^+) \!=\! \mathcal{X}_{1,T}Q(f_1(x))\mathfrak{D}(f_1(x)) \!=\! f_2(x),
	\end{align*}
	and by induction, we have
	\begin{equation*}
		\label{eq:non-k-step}
		x^{k+} = \mathcal{X}_{1,T}Q(f_{k-1}(x))\mathfrak{D}(f_{k-1}(x)) =f_k(x),
	\end{equation*} completing the proof.\hfill $\blacksquare$

\subsection{Proof for Proposition~\ref{thm:data_k-inductive-CBC}}

Proposition~\ref{thm:data_k-inductive-CBC} builds upon the definition of $k$-NBCs found in~\citep{Anand2021kSafety}:

\begin{definition}[$k$-BC~\citep{Anand2021kSafety}]
        We say that a function $\mathcal{B} : X \rightarrow \reals_{\geq0}$ is a
$k$-inductive barrier certificate for the system $\Sigma$ with respect
to a set of initial states $X_{\mathcal I} \subseteq X$ and a set of unsafe states
$X_{\mathcal{U}} \subseteq X$, if there exist $k\in\nats_{>0}$ and $\epsilon\in\reals_{\geq0}$, such
that following conditions hold:
\begin{subequations}
\begin{align}
    &\mathcal B(x) \leq 0, &\forall x \in X_{\mathcal I}, \nonumber\\
&\mathcal B(x) > (k-1)\epsilon, &\forall x \in X_{\mathcal{U}}, \nonumber\\
\label{eq:proof1}
&\mathcal B(f_1(x)) - \mathcal B(x) \leq \epsilon, & \forall x \in X, \\
\label{eq:proof2}
&\mathcal B(f_k(x)) - \mathcal B(x) \leq 0, &\forall x \in X.
\end{align}
\end{subequations}
\end{definition}
Substituting the result of Lemma~\ref{lem:non-Q-matrix} into~\eqref{eq:proof1} and Lemma~\ref{lem:recursive-dt-NS} into~\eqref{eq:proof2} gives conservative certificate conditions. To achieve our further relaxed conditions, we enforce~\eqref{eq:proof1} and~\eqref{eq:proof2} when  $\mathcal B(x) \leq (k-1)\epsilon$ and $\mathcal B(x) \leq 0$, respectively --- enforcing the condition to hold only for states reachable from the initial set $X_\mathcal{I}$.  \hfill$\blacksquare$

\subsection{Proof for Theorem~\ref{thm:k-ind-proof}}

By condition~\eqref{eq:safe_kCBC_cond1} and~\eqref{eq:safe_kCBC_cond2}, the initial and unsafe regions are distinctly separated by two level sets of $\mathcal{B}(x)$, with the initial condition level set \eqref{eq:safe_kCBC_cond1} of lower value than the unsafe region's level set. 

Condition~\eqref{eq:safe_kCBC_cond4} enforces the system to be strictly decreasing in $\mathcal{B}(x)$ value after $k$ steps, while a one step transition may increase by a value of $\epsilon$ \eqref{eq:safe_kCBC_cond3}. The maximal increase in $\mathcal{B}(x)$ value as the system evolves is therefore $k-1$ increases by a value of $\epsilon$. As the unsafe region value is at least $(k-1)\epsilon$, and the largest initial value is $0$, a $k$-BC certificate always guarantees safety. \hfill$\blacksquare$

\section{Linear Systems}
\label{app:linear-systems}

We separate the discrete-time linear system (dt-LS) class from the more general dt-NS class presented in the main paper. For dt-LS, the $k$-step evolution of the system has a simple closed-form solution for any $k$. We consider the following lemma for the dt-LS case, \emph{i.e.} $\mathfrak{D}(x) = x$, to obtain a data-based representation of a closed-loop dt-LS.

\begin{lemma}
	\label{cor:Q-matrix}
	    Let matrix $Q$ be a $(T\times n)$ matrix such that
	    \begin{equation}
		    \label{eq:Q-matrix-X0}
		        \mathbb{I}_n = \mathcal{X}_{0,T}Q,
		    \end{equation}
	    where $\mathcal{X}_{0,T}$ is an $(n\times T)$ full row-rank matrix. The closed-loop system $x^+=Ax$ has the following data-based representation:
	    \begin{equation*}
		        x^+ = \mathcal{X}_{1,T}Qx,~\text{equivalently},~A = \mathcal{X}_{1,T}Q.
		    \end{equation*}
	
	\noindent Additionally, the data-based evolution of the system after $k$ steps is \begin{equation}
		    	\label{eq:k-step}
		    	x^{k+} = (\mathcal{X}_{1}Q)^kx.
		    \end{equation}
	\end{lemma}
\begin{proof}
	The proof is straightforward following the proofs of Lemma~\ref{lem:non-Q-matrix} and Lemma~\ref{lem:recursive-dt-NS}.
	\end{proof}

\begin{remark}
	    Similar to the dt-NS case, to enforce $\mathcal{X}_{0,T}$ to be full row-rank, the number of samples $T$ should be at least $n$.
\end{remark}

For dt-LSs, the $k$-BC conditions leverage Lemma~\ref{cor:Q-matrix}, yielding a simple closed-form expression for the last two conditions.
\begin{proposition}
	\label{cor:dt-LS-barrier-conditions}
	    Given a dt-LS $\Sigma$, a function $\mathcal{B}: X\rightarrow \reals$ is a $k$-inductive barrier certificate ($k$-BC) for $\Sigma$ if there exist $k\in\nats_{>0}$, $\epsilon\in\reals_{\geq0}$, such that
        \begin{subequations}
        \begin{align}
			\mathcal{B}(x)& \leq 0, &\quad\quad\quad\forall x\in X_{\mathcal I}, \\
			\mathcal{B}(x)& > (k-1)\epsilon, &\quad\quad\quad\forall x\in X_{\mathcal U}, \\
			    \label{eq:dtLS_kCBC_cond3}
			    \mathcal{B}(\mathcal{X}_{1}Qx) &\leq \mathcal{B}(x) + \epsilon,\quad &\forall x\in X~\text{where}~\mathcal{B}(x) \leq (k-1)\epsilon,  \\
			    \label{eq:dtLS_kCBC_cond4}
			    \mathcal{B}((\mathcal{X}_{1}Q)^kx) &\leq \mathcal{B}(x) \quad &\forall x\in X~\text{where}~\mathcal{B}(x) \leq 0.
			    \end{align}
		    \end{subequations}
	\end{proposition}
\begin{proof}
The proof is straightforward, following the proof of Proposition~\ref{thm:data_k-inductive-CBC}.
\end{proof}

\begin{SCfigure}[1.2][h!]
	\includegraphics[scale=0.45]{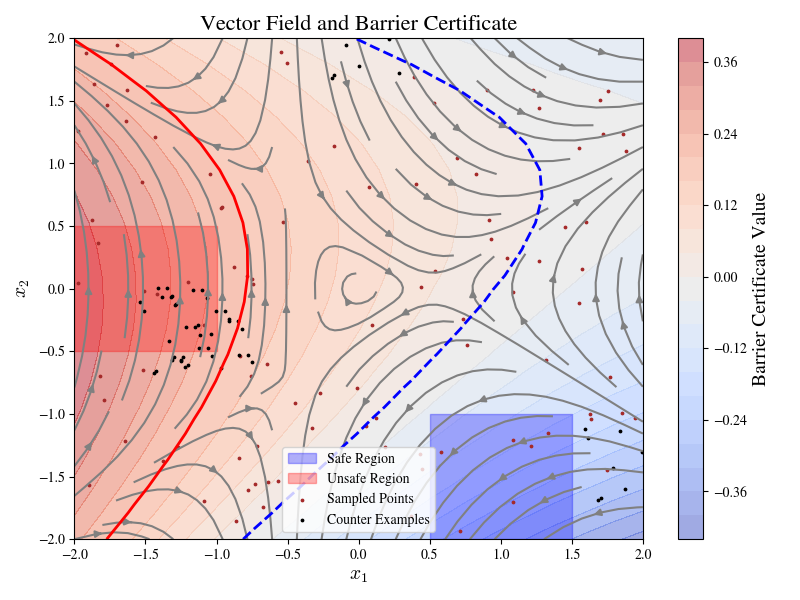}
	\caption{\emph{Verified} $k$-NBC for the \emph{polynomial system}, with the zero level set shown in dashed blue and the unsafe level set in red. Original sampled points are depicted in brown, while counterexamples are in black. The $k$-NBC was identified in the 7th iteration. The contour plot represents barrier certificate values across all states, and the true vector field is illustrated in gray.}\vspace{-1cm}
	\label{fig:poly_ex}
\end{SCfigure}
\vspace{0.5cm}

\section{Additional Case Studies}
\label{app:case-studies}

\subsection{Polynomial System}

\noindent\textbf{System Description.} We first consider a discrete-time polynomial system as
\begin{align*}
	\begin{bmatrix}
		x_1^+\\
		x_2^+
	\end{bmatrix} = \begin{bmatrix}
		x_1\\x_2
	\end{bmatrix} +  \Delta t\begin{bmatrix}
		x_2 + 2x_1x_2\\
		-x_1 + 2x_1^2 -2x_2^2
	\end{bmatrix}\!\!, 
\end{align*}
where the sampling time $\Delta t$ is $0.1$ seconds. The model is assumed to be \emph{unknown}, but an extensive dictionary of terms is available: $\mathfrak{D}(x) = [x_1, x_2, x_1x_2, x_1^2, x_2^2]$. This dictionary is constructed based on a first-principles understanding of the system, including all monomials up to degree 2. We consider a safety specification over the state space $X=[-2,2]^2$ with an initial region $X_\mathcal I = [0.5,1.5]\times[-2,-1]$ and unsafe region $X_\mathcal U = [-2,-1]\times[-0.5,0.5]$. As our dictionary has $5$ terms, we consider a single state trajectory of length $6$, from initial state $x_0 = [0.5,-2]$. Using the datasets $\mathcal{X}_0$, $\mathcal{X}_1$, and $\mathcal{D}_0$, we solve~\eqref{eq:non-Q-matrix-X0}. We now obtain the data-driven model $x^+ = \mathcal{X}_1Q(x)\mathfrak{D}(x)$.

\noindent\textbf{Initial Candidate.} We aim to find a $k$-NBC for this system, with $k=3$ and $\epsilon=0.1$. For conservatism, the parameters are set to $\eta_1 = 0.1$ and $\eta_2 = 0.001$. We sample $100$ data points from the data-driven model to form our dataset triple $(S,S^+,S^{k+})$. We also consider a neural network with two inputs $x_1,x_2$ (one per dimensions), $1$ hidden layer containing $2$ nodes, and one output which is the evaluation of our candidate $k$-NBC. Each node $y$ evaluates the formula \begin{equation}
	\label{eq:activation-function}
	y = (b + \sum_{i=1}^2 x_iw_i)^2,
\end{equation} 
where, for each node: $b$ represents the bias, $w_i$ are the weights, and the activation function is quadratic. For training, backpropagation is performed using the Adam optimizer~\citep{kingma2014adam}, running for $1000$ epochs with a learning rate of $0.1$. Verification with Z3~\citep{z3} identified a counterexample at $[0,\frac{15}{8}]$.

\noindent\textbf{Retraining.} Prompted by this counterexample, it along with 20 points within a 0.1 radius are added to the dataset. The model is then retrained for another $1000$ epochs with a reduced learning rate of $0.05$. This loop was repeated for each counterexample until iteration $7$ where the SMT solver Z3 returned \texttt{unsat} verifying the $k$-NBC $\mathcal{B}(x)$ as valid for all $x$, as seen in Figure~\ref{fig:poly_ex}. In the figure, the true vector field is plotted but all results were generated and verified using the data-driven model. Allowing for rounding, the $k$-NBC is designed as
\[\mathcal{B}(x) = 0.02x_1^2 + 0.02x_1x_2 - 0.12x_1 - 0.04x_2^2 + 0.04x_2 + 0.10\]
which is not sum-of-squares, since its quadratic form is not positive semidefinite.

\noindent\textbf{Conventional Barrier Certificate Solution.} We investigate whether $k$-NBC satisfies the conventional barrier certificate conditions (where $k=1$ and $\epsilon=0$). However, Z3 identifies a counterexample at $[0,\frac{255}{128}]$.

\subsection{Pendulum System}

\begin{SCfigure}[1.2][h]
	\includegraphics[scale=0.5]{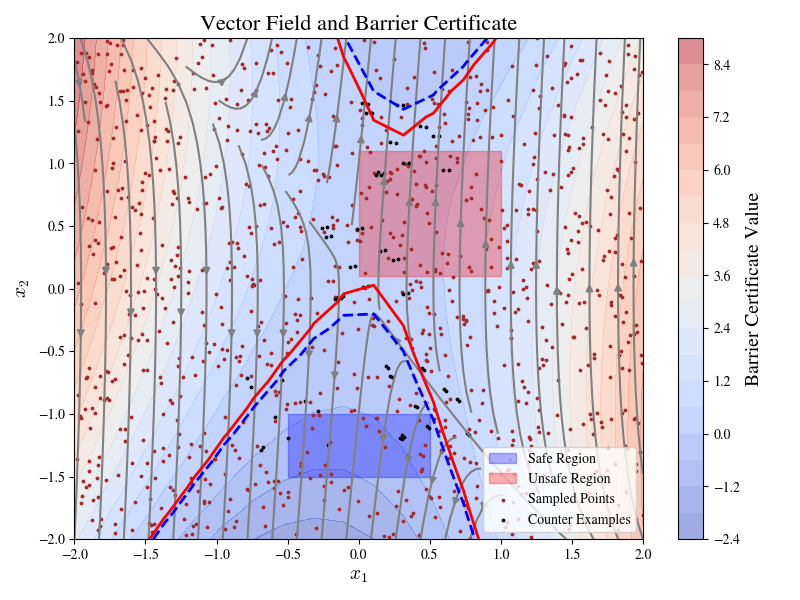}
	\caption{\emph{Verified} $k$-NBC for the \emph{pendulum system}, with the zero level set in dashed blue and the unsafe level set in red. Original sampled points are shown in brown, and counterexamples in black. The $k$-NBC was obtained after one retraining. The contour plot illustrates barrier certificate values across all states, while the true vector field is depicted in gray.}
	\label{fig:nonpoly_ex}\vspace{-0.9cm}
\end{SCfigure}
\vspace{0.5cm}

 \noindent\textbf{System Description.} For our second case study, we consider a controlled pendulum system as
\begin{equation*}
	\begin{bmatrix}
		x_1^+\\x_2^+\end{bmatrix} = \begin{bmatrix}
		x_1\\x_2
	\end{bmatrix} + \Delta t\begin{bmatrix}
		x_2\\
		\text{sin}(x_1)g + -\frac{b}{m}x_2 + \frac{-gmlx_1 + bx_2}{ml}
	\end{bmatrix}\!\!, 
\end{equation*}
where $x_1$ is the pendulum angle, $x_2$ is the angular velocity, $\Delta t = 0.1$ is the sampling time, $g=9.81$ is the gravitational acceleration, $m=1$ is the mass of the pendulum, $l=0.1$ is the pendulum's length, and $b=1.0$ is the damping coefficient.

Again, we consider the system to be \emph{unknown}, but we have an extensive dictionary of terms $\mathfrak{D}(x) = [x_1,x_2,~\text{sin}(x_1),~\text{cos}(x_1)]$. We consider a safety specification with regions $X=[-2,2]^2, X_I = [-0.5,0.5]\times[-1.5,-1]$ and $X_U = [0,1]\times[0.1,1.1]$. As the extensive dictionary has $4$ terms, we take a single state trajectory of length $5$ from initial state $[0.5,-1.5]$.

\noindent\textbf{Initial Candidate.} To find a $k$-NBC, we set $k=2$, $\epsilon=0.1$, $\eta_1 = 0.1$, and $\eta_2 = 0.001$. We sample $1000$ data points to form our dataset triple $(S,S^+,S^{k+})$. We consider a neural network with two inputs $x_1,x_2$ (one per dimensions), $1$ hidden layer containing $32$ nodes with the same description as~\eqref{eq:activation-function}, and one output which is the evaluation of our candidate $k$-NBC. Using the Adam optimizer, we run $1000$ epochs with learning rate $0.1$ and verify candidates using the dReal~\citep{dreal} with $\delta$-bound set to $0.001$. A counterexample is found on the first iteration in the range $[0.4811,0.4813]\times[-1.0002,-1.0000]$.

\noindent\textbf{Retraining.} From the range of the counterexample, $10$ additional points are added to the dataset for retraining of another $1000$ epochs with learning rate $0.05$. The process is repeated until iteration $8$ where the SMT solver returns no counterexample, verifying the $k$-NBC valid for all $x$, as seen in Figure~\ref{fig:nonpoly_ex}. The verified $k$-NBC $\mathcal{B}(x)$ is omitted here due to its length.

\noindent\textbf{Conventional Barrier Certificate Solution.} We investigate whether $k$-NBC satisfies the conventional barrier certificate conditions (where $k=1$ and $\epsilon=0$). However, dReal identifies a counterexample in the range $[0.282,0.284]\times[-0.438,-0.436]$.

\end{document}